\documentclass[onecolumn,draftcls,12pt]{IEEEtran}
\usepackage{amsfonts}
\usepackage{latexsym,bm}
\usepackage{amsmath,amssymb,amsfonts}
\usepackage{indentfirst}
\usepackage{cite}
\usepackage{graphics}
\usepackage{graphicx,color,psfrag,overpic}
\usepackage{stfloats}
\usepackage{cases}
\usepackage{array}
\usepackage{verbatim}
\usepackage{setspace}
\usepackage{fancyhdr}

\usepackage{amsthm}

\newtheorem{theorem}{Theorem}

\def\bG{{\bf G}}
\def\bR{{\bf R}}
\def\bn{{\bf n}}
\begin{document}
\title{Spectral Efficiency of Mixed-ADC Receivers for Massive MIMO Systems}
\vspace{2mm}
\author{\IEEEauthorblockN{Weiqiang Tan,$~\IEEEmembership{Student Member,~IEEE,}$ Shi Jin,$~\IEEEmembership{Member,~IEEE}$, \\Chao-Kai Wen, $\IEEEmembership{Member,~IEEE}$ and Yindi Jing,$~\IEEEmembership{Member,~IEEE}$}
\thanks{W. Tan and S. Jin are with the National Communications Research Laboratory, Southeast University, Nanjing 210096, P. R. China (email: \{wqtan, jinshi\}@seu.edu.cn).}
\thanks{C. Wen is with Institute of Communications Engineering, National Sun Yat-sen University, Kaohsiung 804, Taiwan (email: chaokai.wen@mail.nsysu.edu.tw).}
\thanks{Y. Jing is with the Department of Electrical and Computer Engineering,
University of Alberta, Edmonton, AB, T6G 2V4 Canada (email: yindi@ualberta.ca).}}

\maketitle
\vspace{-5mm}
\begin{abstract}
This paper investigated the uplink of multi-user massive multi-input multi-output (MIMO) systems with a mixed analog-to-digital converter (ADC) receiver
architecture, in which some antennas are equipped with costly full-resolution ADCs and others with less expensive low-resolution ADCs. A closed-form approximation of the achievable spectral efficiency (SE) with the maximum-ratio combining (MRC) detector is derived. Based on this approximated result, the effects of the number of base station antennas, the transmit power, the proportion of full-resolution ADCs in the mixed-ADC structure, and the number of quantization bits of the low-resolution
ADCs are revealed. Results showcase that the achievable SE increases
with the number of BS antennas and quantization bits, and it converges to a saturated value in the high user power regime or the full ADC resolution case. Most important, this work efficiency verifies that for massive MIMO, the mixed-ADC receiver with a small fraction of full-resolution ADCs can have comparable SE performance
with the receiver with all full-resolution ADCs but at a considerably lower hardware cost.
\end{abstract}
\section{Introduction}
Massive multi-input multi-output (MIMO) systems, in which the base station (BS) is equipped with a large number of antennas, have emerged as a promising technology
for 5G systems because of their significant advantages in energy efficiency and spectral efficiency (SE). However, the huge hardware cost of the processing units
associated with the antennas is a major issue in their practical implementation. This motives the use of low-resolution analog-to-digital converters
(ADCs) with favorable low cost and low power consumption, especially when the number of BS antennas is large or the sampling rate is high (e.g., in the GHz range)
\cite{RefLiang,Fan:Lett2015,Bai:2013}.

As recently reported in \cite{Mo:2014}, a massive MIMO receiver with pure low-resolution ADCs suffers from considerable rate loss in the high-SNR regime. When all
ADCs have low resolution, time-frequency synchronization and channel estimation are challenging and require excessive overhead \cite{RefLiang},\cite{Zhang:2016}.
In \cite{RefLiang}, a mixed-ADC receiver architecture was proposed for massive MIMO systems, in which a fraction of the ADCs has full-resolution to promote
system performance, and the others have low-resolution in consideration of the hardware cost and energy consumption. By analyzing the general mutual information \cite{Zhang:2015},
authors demonstrated that such a mixed-ADC architecture can achieve a significant fraction of the channel capacity with a small number of
full-resolution ADCs \cite{RefLiang}. In \cite{Zhang:2016}, the mixed-ADC receiver architecture was developed for multiuser massive MIMO systems, in which a family of Bayes detectors were developed by conducting Bayesian estimate of the user signals with GAMP algorithm. As massive MIMO is a newborn technology, studying on the achievable spectral efficiency (SE) in massive MIMO with a mixed ADC
receiver architecture is of special interest.

In this paper, we aim to analyze the uplink achievable SE of multi-user  massive MIMO systems with a mixed-ADC receiver architecture. We adopt the low-complexity
maximum-ratio combining (MRC) detection and derive a closed-form approximated expression of the achievable SE. The result reveals the effects of important system
parameters, including the number of BS antennas, the user transmit power, the number of quantization bits of the low-resolution ADCs, and the proportion of
full-resolution ADCs in the mixed structure, on the achievable SE. Simulations show that the derived approximation is tight with the transmitter power $p_u < 10$ dB or
when the number of quantization bits is moderate to high. Our work also shows that the mixed-ADC massive MIMO receiver with a small fraction of
full-resolution ADCs can have comparable SE performance with the receiver with all full-resolution ADCs.

\section{System model}
We consider a single-cell multi-user massive MIMO uplink with $N$ users and $M$ ($M\!\gg \!N$) receiver antennas at the BS. A mixed-ADC architecture is used at the
receiver to save on hardware cost. The receiver has  $M$ antenna elements. However, only $M_0$ antennas are connected to the costly full-resolution ADCs, and the
remaining $M_1$ (where $M_1\!=\! M\!-\!M_0$) antennas are connected to the less expensive low-resolution ADCs. Define $\kappa\triangleq M_0/M $ (0$ \leq \kappa \leq 1$),
which is the proportion of the full-resolution ADCs in the mixed-ADC architecture.

Let $\bG$ be the $M\times N$ channel matrix from the users to the BS. The channel matrix is modeled as
\begin{equation}\label{eq:one02}
{\bf{G}} = {\bf{H}}{{\bf{D}}^{{1 \mathord{\left/
 {\vphantom {1 2}} \right.
 \kern-\nulldelimiterspace} 2}}},
\end{equation}
where ${\bf{H}}\! \in\! \mathbb{C}^{M \!\times \!N}$ contains the fast-fading coefficients, whose entries that are independent and identically distributed (i.i.d.)
complex Gaussian random variables with zero-mean and unit-variance denoted $\mathcal{C N}(0,1)$, and ${\bf{D}}$ is an $N \!\times\! N$ diagonal matrix with diagonal
elements given by ${[\bf{D}]}_{nn}\! =\! {\beta _n}$. Here, ${\beta _n} =  {{{z_n}}}{r_{n}^{-\gamma }}$ models both path loss and shadowing, where $r_n$ is the
distance from the  $n$th user to the BS,  $\gamma$ is the decay exponent, and $z_n$ is a log-normal random variable. For ease of expression, we denote $\bG=[\bG_0\
\bG_1]^T$, where $\bG_0$ is the $M_0 \times N $ channel matrix from the users to the $M_0$ BS antennas with full-resolution ADCs, and $\bG_1$ is the $M_1\!\times\! N$
channel matrix from the users to the $M_1$ BS antennas with low-resolution ADCs.

Let $p_u$ be the average transmitted power of each user and $\bf{x}$ be the $N \times 1$ vector of information symbols. The received signals at the full-resolution ADCs can be given by
\begin{equation}\label{ttt:one01}
{\bf y}_0 = \sqrt {{p_u}} {\bf G}_0 {\bf x} + {\bf n}_0,
\end{equation}
where ${{\bf{n}}_0 \sim \mathcal{C N}\left( {0, {\bf{I}}} \right)}$ is the additive white Gaussian noise (AWGN). The received signals at low-resolution ADCs are
quantized each with a $b$-bit scalar quantizer. In general, the quantization operation is nonlinear, and the quantization error is correlated with the input
signal. For a tractable analysis, we utilize the additive quantization noise model \cite{Bai:2013} \cite{Zhang:2015}  that enables standard linear processing and Gaussian decoding at
the decoder. This approximation is widely used in quantized MIMO systems and shows to have good accuracy\cite{Fan:Lett2015,Bai:2013}. The received signals through
low-resolution ADCs can be approximated by
\begin{equation}\label{ttt:one02}
{{\bf y}_1} = Q\left( { {{\bf{\tilde{y}}}_1}} \right) \approx  \alpha {{\bf{\tilde{y}}}_1} + {\bf{n}}_{\sf q}=\sqrt {{p_u}} \alpha {{\bf{G}}_1}{\bf{x}} + \alpha {{\bf{n}}_1} + {\bf n}_{\sf q},
\end{equation}
where ${{\bf{n}}_1 \sim \mathcal{C N}\left( {0, {\bf{I}}} \right)}$ is the AWGN, ${{\bf{n}}_{\sf q}}$ is the additive Gaussian quantization noise vector that is
uncorrelated with ${{\bf{y}}_1}$, and $\alpha$ is a coefficient. For the non-uniform scalar minimum-mean-square-error (MMSE) quantizer of a Gaussian random
variable, the values of $\alpha$ are listed in Table \ref{Table} for $b \leq 5$. For $b > 5$, the values of $\alpha$ can be approximated by $1 - {\textstyle{{\pi \sqrt 3 } \over 2}} {2^{{\rm{ - }}2b}}$  \cite{Fan:Lett2015}. Note that $\alpha=1 $ as $b \rightarrow \infty$. From
\eqref{ttt:one01} and \eqref{ttt:one02}, the overall received signals at the BS can be expressed as
\begin{equation}\label{eq:one04}
{\bf{y}} = \left[ \begin{gathered}
  {{\bf{y}}_0} \hfill \\
  {{\bf{y}}_1} \hfill \\
\end{gathered}  \right] \approx \left[ \begin{gathered}
  \sqrt {{p_u}} {{\bf{G}}_0}{\bf{x}} + {\bf{n}}_0 \hfill \\
  \sqrt {{p_u}} \alpha {{\bf{G}}_1}{\bf{x}} + \alpha {\bf{n}}_1 + {\bf{n}}_{\sf q} \hfill \\
\end{gathered}  \right].
\end{equation}

We assume that the BS has perfect channel state information (CSI), which can potentially be obtained, e.g., through exploiting uplink channel feedback in frequency division duplexing systems or channel reciprocity in time division duplex systems. Furthermore, the BS adopts the MRC linear detector, which has significantly lower computational complexity than
other linear detectors such as the zero-forcing and linear MMSE estimators. By using \eqref{eq:one04}, the received signal vector after the MRC combination is
given by
\begin{eqnarray}
&&\hspace{0mm}{\bf{r}} = {{\bf{G}}^H}{\bf{y}} = {\left[ \begin{gathered}
  {{\bf{G}}_0} \hfill \\
  {{\bf{G}}_1} \hfill \nonumber \\
\end{gathered}  \right]^H}\left[ \begin{gathered}
  \sqrt {{p_u}} {{\bf{G}}_0}{\bf{x}} + {{\bf{n}}_{{0}}} \hfill \\
  \sqrt {{p_u}} \alpha {{\bf{G}}_1}{\bf{x}} + \alpha {{\bf{n}}_{{1}}} + {{\bf{n}}_{{\sf q}}} \hfill \\
\end{gathered}  \right],
\end{eqnarray}
which can be further expressed as
\begin{equation}\label{eq:one09}
\!{\bf{r}}\! = \!\sqrt {{p_u}} \left( {{\bf{G}}_0^H{{\bf{G}}_0}\! + \alpha {\bf{G}}_1^H{{\bf{G}}_1}} \right){\bf{x}} \!+\! \left( \!{{\bf{G}}_0^H{{\bf n}_0} \!+ \!\alpha {\bf{G}}_1^H{{\bf n}_1}} \!\right)\! +  {\bf{G}}_1^H{{\bf{n}}_{\sf q}}.
\end{equation}
\begin{table}
\centering
\caption{{$\alpha$ for different ADCs quantization bits ($b \leq 5$).   }}
\begin{tabular}{|c|c|c|c|c|c|c|c|c|c|c|} \hline
    $ {b}  $ & $1$ &$2$ & $3$ & $4$  & $5$ \\
    \hline
    $ { \alpha}$ & $0.6366$ & $0.8825$ & $0.96546$ & $0.990503$  & $0.997501$\\
    \hline
\end{tabular}
\label{Table}
\end{table}
From \eqref{eq:one09}, the $n$th element of ${\bf{r}}$ can be expressed as
\begin{equation}\label{eq:one10}
{r_n} = \sqrt {{p_u}} \left( {{\bf{g}}_{n0}^H{{\bf{g}}_{n0}} \!+\! \alpha {\bf{g}}_{n1}^H{{\bf{g}}_{n1}}} \right){x_n} + \sqrt {{p_u}} \sum\limits_{i = 1,i \ne n}^N \!\! \left( {{\bf{g}}_{n0}^H{{\bf{g}}_{i0}} + \alpha {\bf{g}}_{n1}^H{{\bf{g}}_{i1}}} \right){x_i} + \left( {{\bf{g}}_{n0}^H{{\bf{n}}_0} + \alpha {\bf{g}}_{n1}^H}{{\bf{n}}_1} \right) + {\bf{g}}_{n1}^H{{\bf{n}}_{\sf q}},
\end{equation}
where ${{\bf{g}}_{n0}}$  and ${{\bf{g}}_{n1}}$ are the $n$th columns of ${\bf{G}}_{0}$ and ${\bf{G}}_{1}$, respectively.

\section{Achievable SE Analysis}
In this section, we derive an closed-form approximated expression for the achievable SE of the multi-user massive MIMO uplink with the mixed-ADC receiver architecture. Based on this derived SE result, the behaviors of the system SE with respect to several parameters are revealed including the number of BS antennas, the user transmit power, the number of quantization bits of the low-resolution ADCs, and the ratio of full-resolution ADCs in the mixed-ADC architecture.

According to \eqref{eq:one10}, with straightforward calculations, the achievable SE of the $n$th user can be expressed by
\begin{equation}\label{eq:one11}
\!\!{R_n} \hspace{-0.5mm}= \hspace{-0.5mm}{\mathbb E} \left\{ {{{\log }_2}\left( \!\!{1  \!\!+  \! \! \frac{{{p_u}{ {\left| {{\bf{g}}_{n0}^H{{\bf{g}}_{n0}}} \!+ \!\hspace{-1mm}{\alpha}{{\bf{g}}_{n1}^H{{\bf{g}}_{n1}}} \right|}^2}}}{{ {I_n} \!+ \! {{\left| {{{\bf{g}}_{n0}^{H}}}{{\bf{n}}_0}  \hspace{-0.5mm}  \!+ \!\!{\alpha}  {{{\bf{g}}_{n1}^{H}}}{{\bf{n}}_1} \right|}}^2\!+ \!\left| {{\bf{g}}_{n1}^H{{\bf{n}}_{\sf q}}} \right|^2 }}} \!\right)}\! \right\},
\end{equation}
where $|\cdot|$ stands for the absolute value operation and
\begin{align}\label{eq:one12}
I_n =  {p_u}\! \sum\limits_{i = 1,i \ne n}^N \! {  {\left| {{\bf{g}}_{n0}^H{{\bf{g}}_{i0}}} + {\alpha} {{\bf{g}}_{n1}^H{{\bf{g}}_{1i}}} \right|}^2},
\end{align}
which is the interference power.

For a rigorous analysis of the achievable SE, the probability density function (pdf) of the SINR, which is the second term inside the log-function in \eqref{eq:one11}, is needed. This rigorous analysis is highly challenging, if not impossible, because of the complicated expression. Alternatively, we seek for a tight
approximation of the achievable SE by taking advantage of the large-scale of the massive MIMO syetem. Our main result on the achievable SE is presented in the following
theorem.
\begin{theorem}\label{theorem01}
For a multi-user massive MIMO uplink with MRC and a mixed-ADC receiver, under perfect CSI, the achievable SE of the $n$th user can be approximated as follows:
\setlength{\arraycolsep}{1pt}
\begin{eqnarray}\label{eq:one25}
{{\tilde R}_n} \!\triangleq  \!{\log _2}\hspace{-1mm}\left(1 + \!M\frac{\left[  \alpha  + (1-\alpha) \kappa \right]\beta _n }{\frac{1}{{{p_u}}} +  \sum\limits_{i = 1,i \ne n}^N  \beta _i + \frac{2\alpha \left( 1 - \alpha  \right)(1 - \kappa ) }{\alpha  + (1-\alpha) \kappa} \beta _n} \right)\hspace{-1mm}.
\end{eqnarray}
\setlength{\arraycolsep}{5pt}
\end{theorem}
\begin{proof}
For a large but finite $M$, we can use the approximation \cite{Zhang:2014}: $\mathbb E\left\{ {{{\log }_2}\left( {1 + \frac{X}{Y}} \right)} \right\} \approx {\log
_2}\left( {1 + \frac{{\mathbb E\left\{ X \right\}}}{{\mathbb E\left\{ Y \right\}}}} \right)$ for independent random variables $X$ and $Y$ that converge to their
means with probability one as $M\rightarrow \infty$. Thus, $R_n$ can be approximated by
\begin{align}\label{eq:one15}
{\tilde R_n} \!= \!{\log _2}\hspace{-1mm} \left( \!{1\! + \!\frac{{{p_u}\left( {{\mathbb E}\! \left\{ {{{\left| {{\bf{g}}_{n0}^H{{\bf{g}}_{n0}}} \right|}^2}} \right\} + 2\alpha {\mathbb E}\left\{  {{\bf{g}}_{n0}^H{{\bf{g}}_{n0}}}  \right\}{\mathbb E} \!\left\{   {{\bf{g}}_{n1}^H{{\bf{g}}_{n1}}}   \right\}\! +\! {\alpha ^2}{\mathbb E} \left\{ {{{\left| {{\bf{g}}_{n1}^H{{\bf{g}}_{n1}}} \right|}^2}} \right\}} \right)}}{{{ \tilde  I_n} +   {\mathbb E} \left\{ {{{\left\| {{{\bf{g}}_{n0}}} \right\|}^2}} \right\} +  2\alpha {\mathbb E} \left\{ {\left\| {{{\bf{g}}_{n0}}} \right\|} \right\}{\mathbb E} \left\{ {\left\| {{{\bf{g}}_{n1}}} \right\|} \right\}  +   {\alpha ^2}{\mathbb E} \left\{ {{{{\left\| {{{\bf{g}}_{n1}}} \right\|}}^2}} \right\}+ {\mathbb E} \left\{ \left| {{\bf{g}}_{n1}^H{{\bf{n}}_{\sf q}}} \right|^2 \right\} }}}\! \right),
\end{align}
where $\|\cdot\|$ stands for Euclidean norm and
\begin{align}\label{eq:one15}
\tilde I_n  = {p_u}\! \sum\limits_{i = 1,i \ne n}^N \!{\left( {{\mathbb E} \left\{ {{{\left| {{\bf{g}}_{n0}^H{{\bf{g}}_{i0}}} \right|}^2}} \right\} + 2\alpha {\mathbb E}\left\{   {{\bf{g}}_{n0}^H{{\bf{g}}_{i0}}}   \right\}{\mathbb E} \left\{  {{\bf{g}}_{n1}^H{{\bf{g}}_{i1}}} \right\}\! + \!{\alpha ^2}{\mathbb E} \left\{ {{{\left| {{\bf{g}}_{n1}^H{{\bf{g}}_{i1}}} \right|}^2}} \right\}} \right)}.
\end{align}

Next, we calculate the expectation of each term. From the channel model in (\ref{eq:one02}), we have
\begin{equation}
{{\bf{g}}_{nj}} = \sqrt {{\beta _n}} {{\bf{h}}_{nj}},
\end{equation}
where the entries of ${\bf{h}}_{nj}$ are i.i.d. following $\mathcal{CN}(0,1)$ and the subscript $j$ equals to 0 or 1. Moreover, ${\bf{h}}_{nj}$ and ${\bf{h}}_{ij}$ are independent for $n \ne i$. Thus,
with straightforward calculations, we have ${\mathbb E} \left\{ {{{\left| {{\bf{g}}_{nj}^H{{\bf{g}}_{nj}}} \right|}^2}} \right\}\!\! =\!\! \beta _n^2\left( {M_j^2 + {M_j}} \right)$, ${\mathbb E}
\left\{ {{  {{\bf{g}}_{nj}^H{{\bf{g}}_{nj}}}  }} \right\} \!\!= \!\!{\beta _n}{M_j}$, ${\mathbb E} \left\{ {{  {{\bf{g}}_{nj}^H{{\bf{g}}_{ij}}}
 }} \right\} = 0 $, and ${\mathbb E} \left\{ {{{\left| {{\bf{g}}_{nj}^H{{\bf{g}}_{ij}}} \right|}^2}} \right\} = {\beta _n}{\beta _i}{M_j}$.
The
quantization noise term in \eqref{eq:one15} can be calculated as follows:
\begin{equation}\label{eq:two13}
{\mathbb E} {\left\{ \left| {{\bf{g}}_{n1}^H{{\bf{n}}_{\sf q}}} \right|^2 \right\}}
= {\mathbb E} {\left\{ \left| {\bf{g}}_{n1}^H {{\bf{R}}_{{{\bf{n}}_{{\sf q}}}{{\bf{n}}_{{\sf q}}}}}{\bf{g}}_{n1}\right| \right\}},
\end{equation}
where $\bR_{{\bn}_{\sf q}}$ is the covariance matrix of ${\bn}_{\sf q}$. Following the results in \cite{Fan:Lett2015}, we have
\begin{align} \label{eq:two14}
{\mathbb E} \left\{ \left| {{\bf{g}}_{n1}^H{{\bf{n}}_{\sf q}}} \right|^2 \right\}\!  = \!\alpha(1\!-\! \alpha) M_1{\left(\! {{\beta _n} \!+ \!{p_u}{\beta _n}\!\!\sum\limits_{i = 1}^N {{\beta _i}} \! +\! {p_u}\beta _n^2} \right)}.
\end{align}
Substituting the above calculation results into \eqref{eq:one15}, ${{\tilde R}_n}$ can be written as
{\small{\begin{align}\label{eq:one150505}
{{\tilde R}_n} = {\log _2}\left( {1 + \frac{{{p_u}\left( {\beta _n^2\left( {M_0^2 + {M_0}} \right) + 2\alpha \beta _n^2{M_0}{M_1} + {\alpha ^2}\beta _n^2\left[ {{{\left( {M - {M_0}} \right)}^2} + \left( {M - {M_0}} \right)} \right]} \right)}}{{{p_u}\sum\limits_{i = 1,i \ne n}^N {\left( {{\beta _n}{\beta _i}{M_0} + {\alpha ^2}{\beta _n}{\beta _i}{M_1}} \right)}  + {\beta _n}{M_0} + {\beta _n}{\alpha ^2}{M_1} + \alpha \left( {1 - \alpha } \right){M_1}\left( {{\beta _n} + {p_u}{\beta _n}\sum\limits_{i = 1}^N {{\beta _i}}  + {p_u}\beta _n^2} \right)}}} \right).
\end{align}}}
After some basic manipulations, ${{\tilde R}_n}$ can be simplified as the desired result.
\end{proof}


From  {\emph{Theorem}} 1, it is obvious that the achievable SE in \eqref{eq:one25} is a function of the total number of antennas $M$, user transmit power $p_u$, the proportion of the number of full-resolution ADCs in the mixed-ADC structure $\kappa$, and the number of quantization bits $b$. We see that he SE increases as the user transmit power and the total number of antennas regardless of the values of $\kappa$ and $\alpha$, which is a natural phenomenon consistent with the massive MIMO system. In addition, we observe that $b$ affects the SE through coefficient $\alpha$. The value of $\alpha$ increases as $b$ increases, which boosts the achievable SE. When $b \to \infty$, the value of $\alpha$ converges to one, in which case the quantization error brought by ADCs disappears, and $\kappa$ affects on both the signal power and the quantization noise power because it appears in both the numerator and the denominator of the SINR. A higher $\kappa$ not only increases the desired signal but also reduces the quantization noise effect, which improves the achievable SE.

We start by investigating the impact of the proportion of the full-resolution ADCs on the approximation achievable SE. By differentiating the SINR term in \eqref{eq:one25} with respect to $\kappa$ and after straightforward manipulations, we obtain ${{\partial{\text{SINR}}}}/{{\partial{\kappa}}}
> 0$. That is, the $\textrm{SINR}$ is a monotonically increasing function of $\kappa$. This implies that increasing the number of full-resolution ADCs
benefits the achievable SE. However, more full-resolution ADCs means higher hardware cost and power consumption. Our result in (\ref{eq:one25}) reveals the natural
tradeoff between performance and cost for a massive MIMO receiver. The following results are analyzed the two special cases of $\kappa = 1$ and $\kappa = 0$. For the case of $\kappa = 1$, ${{\tilde R}_n}$ in (\ref{eq:one25}) becomes
\begin{equation}\label{eq:two01}
\tilde R_n = \!{\log _2}\hspace{-1mm}\left( {1+  M\frac{{  {\beta _n}}}{{\frac{1}{{{p_u}}}  +  \!\sum\nolimits_{i = 1,i \ne n}^N   {{\beta _i}}\!}}} \right).
\end{equation}

When $\kappa=1$, all ADCs have a full resolution, which maximums hardware cost and power consumption of system. This case considered in most existing works on massive MIMO. Our SE result in \eqref{eq:two01} is consistent with the previous result in \cite{Zhang:2014} for a full-resolution ADC massive MIMO receiver.
Moreover, the achievable SE for the $\kappa=1$ case is the same as that in \eqref{eq:one25}
for $\alpha = 1$. Both indicate the case in which all BS antennas are equipped with full-resolution ADCs, thus having high hardware cost and power consumption. We now consider the case of $\kappa = 1$. From (\ref{eq:one25}), we obtain
\begin{equation}\label{eq:two02}
{\tilde R_n} \!= \!{\log _2}\hspace{-1mm} \left( \!\!{1 + M\frac{{  \alpha {\beta _n}}}{\frac{1}{{{p_u}}}  + {\! \sum\nolimits_{i = 1,i \ne n}^N
\!{{\beta _i}} + 2 \left( {1  - \alpha } \right){\beta _n}}}}\!\right).
\end{equation}

When $\kappa=0$, all ADCs have a low resolution, which brings loss of the achievable SE due to  the quantization error caused by low resolution ADCs. Furthermore, the result in \eqref{eq:two02} is consistent with the conclusion in \cite{Fan:Lett2015}. Now we consider the case of high transmitter power regime, ${{\tilde R}_n}$ becomes
\begin{equation}\label{eq:two02222}
\mathop {\lim }\limits_{p_u \to \infty} \!\!\!{\tilde R_n}\! = \!{\log _2}\hspace{-0.5mm}\left( \!{1 +M  \frac{{\left[  \alpha  + (1-\alpha) \kappa \right]\beta _n }}{{\sum\nolimits_{i = 1,i \ne n}^N {{\beta _i}}  + \frac{{2\alpha \left( {1 - \alpha } \right)\left( {1 - \kappa } \right)}}{{\kappa  + \alpha \left( {1 - \kappa } \right)}}{\beta _n}}}} \!\right).
\end{equation}

This shows that as the transmit power grows without bound, the SINR and achievable SE are limited and converge to finite values. The reason is that the system becomes interference-limited in the high transmit power regime. Both the signal power and the interference power increase as the user power increases, whereas the noise effects diminish.

\section{Numerical Results}
In this section, we provide simulation results on the total uplink achievable SE, which is defined as $R \!\!=\! \!\sum\nolimits_{n = 1}^N {{R_n}} $. The
cell radius in our simulation is set as $r_c \!\!=\!\!1000$m, the guard zone radius is $r_h\!=\!100$m, the decay exponent is  $\gamma \!\!=\!\!2.1$, and the shadow-fading
standard deviation is $\sigma_{\textrm{shad}}\!\!=\!\!4.9 ~\textrm{dB}$. We assume 10 users, which are uniformly distributed in a round cell. The coefficient of each user's
large-fading $\beta_i$ (i.e., $i=1,\dots,10$) is randomly generated as follows: $\{13.13,~6.49,~11.01,~4.87,~29.00,~8.69,~50.02,~96.00,~1.24, ~41.04 \} {\times
10^{-4}}$.

\begin{figure}
\centering
\includegraphics [width=120mm, height=90mm]{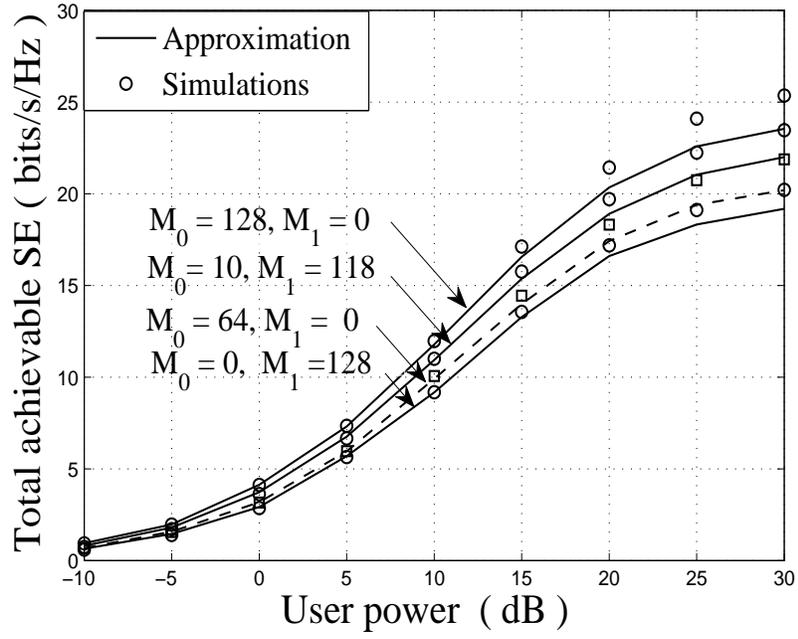}
\caption{Total achievable SE versus user power for $N= 10$ and $b = 4$.
Four cases are considered: 1) $M_0$ = 128, $M_1$ = 0; 2) $M_0$ = 10, $M_1$ = 118; 3) $M_0$ = 0, $M_1$ = 128; and 4) $M_0$ = 64, $M_1$ = 0.
}
\label{fig:fig001}
\end{figure}
\begin{figure}
\centering
\includegraphics [width=120mm, height=90mm]{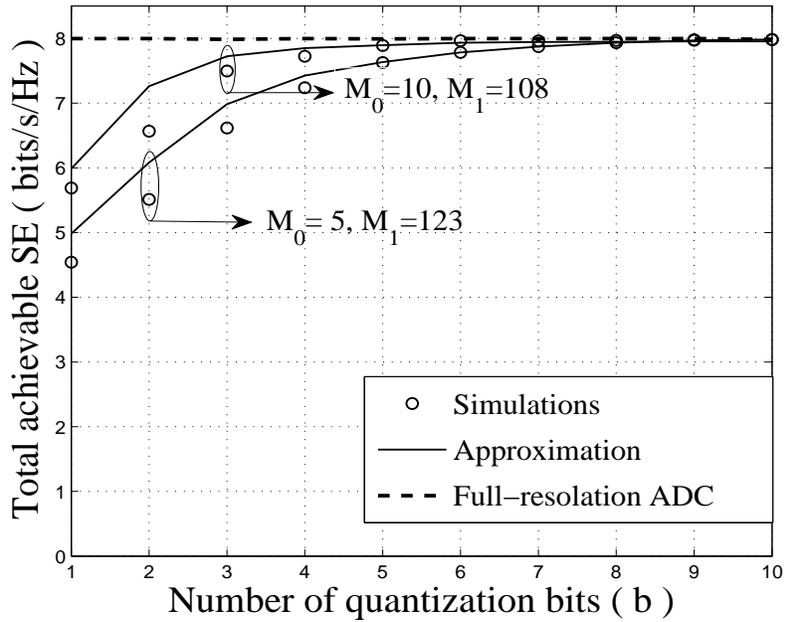}
\caption{Total uplink achievable SE versus the number of quantization bits for $N = 10$ and ${p_u} = 5$ dB.}
\label{fig:fig002}
\end{figure}


In Fig. \ref{fig:fig001}, the simulated total uplink achievable SE in \eqref{eq:one11} is compared with its approximation in \eqref{eq:one25}. Four cases are
considered: 1) $M_0$ = 128, $M_1$ = 0; 2) $M_0$ = 10, $M_1$ = 118; 3) $M_0$ = 0, $M_1$ = 128; and 4) $M_0$ = 64, $M_1$ = 0. We can see from the figure that the approximated formula is very tight
in the the transmitter power $p_u\! < \!10$ dB regime. The SE of all cases are close to each other SE in the low power regime, and increases with the user transmit
power. However, their gap becomes relatively obvious in the high power regime. As expected, the system with full-resolution ADCs has superior SE performance to
that with a mixed-ADC architecture. Compared with Case 1, the SE of Case 2 has only 4.5$\%$ loss at $p_u \!=\! 30$ dB, but Case 2 is economical as the receiver
only uses 10 full-resolution ADCs and the remaining 118 ADCs have low-resolution, thus greatly reducing the hardware cost. Case 3 has poor performance because it
adopts pure low-resolution ADCs. We choose to show the performance of Case 4 in this figure since it as a similar hardware cost as Case 2. The simulation shows that its achievable SE is relatively worse because the total number of BS
antennas has been reduced by half. Moreover, we discuss the energy consumption of different cases. With the energy consumption models $P_{\textrm{full}} = 0.43 M_0 $ Watt for full-resolution
ADCs and $P_{\textrm{low}} = {c_0}{2^b} M_1 + {c_1}$ for low-resolution ADCs,  where $b_{\textrm{max}}$ = 12, ${c_0 = 10^{-4} }$ Watt and ${ {c_1} = 0.02 }$ Watt \cite{Bai:2013}, the total
power consumptions of the four cases are 55.04, 4.4908, 0.2068, and 27.52 Watt, respectively. Comparing the power consumption between Case 1 and Case 2, we see a
great potential of the mixed-ADC architecture for effective power.

\begin{figure}
\centering
\includegraphics [width=120mm, height=90mm]{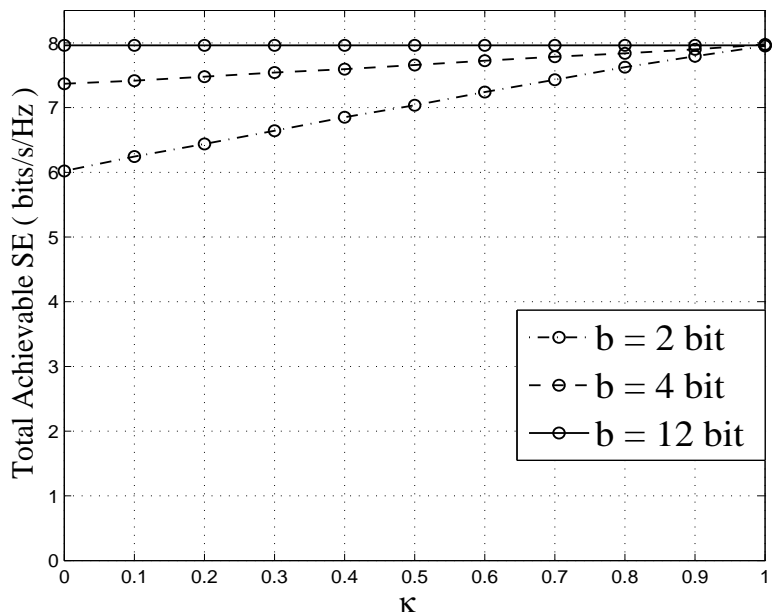}
\caption{Total uplink achievable SE versus $\kappa$ for fixed $M = 128$, $N = 10$ and ${p_u}$ = 5~dB.}
\label{fig:fig003}
\end{figure}

Fig.~\ref{fig:fig002} presents the approximated results for the total achievable SE and the numerical results for various numbers of quantization bits. The
approximation is relatively loose when the quantization bit of the low-resolution ADC is small, i.e., $b=2$. This is because the adopted additive quantization noise model is not tight for small $b$. Moreover, the figure shows that the achievable SE increases with the number of quantization bits and converge
to a limited SE with full-resolution ADCs.


Fig.~\ref{fig:fig003} depicts the approximated results for the total achievable SE against $\kappa$ when $b$ = 2, 4, and 12 bit, respectively. The total achievable
SE increases with $\kappa$. The increasing speed of the achievable SE is much faster for smaller $b$. This implies that the achievable SE can be
improved by simply increasing the number of antennas with full-resolution ADCs for small $b$. For large $b$, i.e., $b=12$, the total achievable SE remains constant as the
quantization error is negligible.

\section{Conclusion}
This paper investigated the uplink achievable SE of a massive MIMO mixed-ADC architecture and an MRC detector. Based on the additive quantization noise model for the ADC quantization, a closed-form approximation on the achievable SE was derived. Our results revealed that the approximated formula is very tight with the transmitter power $p_u < 10$ dB regimes. Moreover, the achievable SE increases
with the number of BS antennas and quantization bits, and it converges to a saturated value as the user power increases or the ADC resolution increases. We conclude that
the mixed-ADC structure can bring most of the desired performance enjoyed by massive MIMO receivers with full perfect-resolution ADCs, but with significantly reduced hardware cost.

\end{document}